\documentclass[12pt]{article}

\usepackage{amsfonts,latexsym,amsthm,amssymb,amsmath,amscd,euscript,tikz,mathtools}
\usepackage{framed}
\usepackage{fullpage}
\usepackage{hyperref}
\hypersetup{colorlinks=true,citecolor=blue,urlcolor =black,linkbordercolor={1 0 0}}
\usepackage{enumitem}
\usepackage{siunitx}
\usepackage{textcomp}
\usepackage{physics}
\usepackage{mathrsfs}
\usepackage{dsfont}
\usepackage[ruled,vlined]{algorithm2e}
\usepackage{cleveref}

\usepackage[authoryear,round]{natbib}
\usepackage{bibentry}
\usepackage{autonum} 
\usepackage{float}

\allowdisplaybreaks[1]

\newtheorem{theorem}{Theorem}

\newtheorem{lemma}{Lemma}
\newtheorem{corollary}{Corollary}

\theoremstyle{definition}
\newtheorem{definition}{Definition}

\theoremstyle{remark}

\newcommand{\cA}{\mathcal{A}}

\newcommand{\cM}{\mathcal{M}}

\newcommand{\cT}{\mathcal{T}}

\newcommand{\bR}{\mathbb{R}}

\newcommand{\choice}{c}
\newcommand{\wpref}{\succeq}

\newcommand{\out}{\text{out}}
\newcommand{\ans}{\cA}
\newcommand{\minans}{\overline{\cA}}

\newcommand{\nc}{\newcommand}

\nc{\on}{\operatorname}
\nc{\Spec}{\on{Spec}}
\nc{\Aut}{\textit{Aut}}
\nc{\id}{\textit{id}}
\nc{\chr}{\on{char}}
\nc{\im}{\on{im}}
\nc{\Hom}{\on{Hom}}
\nc{\lcm}{\on{lcm}}
\nc{\dual}[1]{\prescript{t}{}{#1}}
\nc{\transpose}[1]{{#1}^{\intercal}}
\nc{\Sym}{\on{Sym}}
\nc{\End}{\on{End}}
\nc{\stab}{\on{stab}}
\nc{\Li}{\on{Li}}
\nc{\spn}{\on{span}}
\nc{\sgn}{\on{sgn}}
\nc{\supp}{\on{supp}}
\nc{\Unif}{\on{Unif}}

\usepackage{comment,environ}
\newif\ifnocount

\nocounttrue

\ifnocount

\RenewEnviron{equation}{%
  \setbox0\vbox{%
    $$\refstepcounter{equation}%
    \let\zz\write
    \def\write{\immediate\zz}%
    \BODY$$}%
  \par
  \textbf{Equation}%
  \par}

\catcode`$=\active
\makeatletter
\def${
  \killA%
}
\makeatother
\def\killA#1${\textbf{inline}} 

\else
\fi

\title{On the Computational Properties of Obviously Strategy-Proof Mechanisms\thanks{We thank Mohammad Akbarpour, Ben Golub, Yannai Gonczarowski, Scott Kominers, Chiara Margaria, and Matthew Rabin for helpful comments.  Golowich gratefully acknowledges the support of the Harvard College Herchel Smith Fellowship. All errors remain our own.}}
\author{Louis Golowich\thanks{Harvard University. Email: \url{lgolowich@college.harvard.edu}}\hspace{2mm} and Shengwu Li\thanks{Harvard University. Email: \url{shengwu_li@fas.harvard.edu}}}

\begin{document}

\ifnocount
\else
\maketitle

\begin{abstract}
  We present a polynomial-time algorithm that determines, given some choice rule, whether there exists an obviously strategy-proof mechanism for that choice rule.
\end{abstract}
\fi

\section{Introduction}

Classical mechanism design focuses on direct mechanisms, but there is strong empirical evidence that the extensive form of the mechanism affects the rate of optimal play.  A growing body of laboratory experiments documents that human subjects make frequent mistakes in direct mechanisms, but not in equivalent extensive-form mechanisms \citep{kagel1987information, kagel2001behavior, mcgee2019obvious, ZL2020, schneider2020effects, bo2020iterative, bo2020pick}.

Which mechanisms are easy for human beings to understand?  Some intuitively simple extensive-form mechanisms, such as ascending auctions or dynamic serial dictatorships, are \textit{obviously strategy-proof} \citep{li_obviously_2017}.  Obvious strategy-proofness requires that for any information set reachable under the truth-telling strategy, and for any deviation starting at that information set, every possible outcome under the deviation is no better than every possible outcome from truth-telling.  This formalizes a cognitive limitation: A mechanism is obviously strategy-proof if and only if participants can detect that truth-telling is optimal without using contingent reasoning.\footnote{For laboratory evidence of errors in contingent reasoning, see \cite{charness2009origin}, \cite{esponda2014hypothetical}, \cite{ngangoue2015learning}, and \cite{esponda2016contingent}.}  Another interpretation of obviously strategy-proof mechanisms is that they provide incentives under weaker assumptions about the designer's commitment power.

Obvious strategy-proofness depends on the extensive form of the mechanism.  This raises new obstacles for mechanism design.  By the revelation principle, the question ``does there exist a strategy-proof mechanism for choice rule $c$?'' reduces to the question of whether the corresponding direct mechanism is strategy-proof \citep{gibbard1973manipulation, myerson1979incentive}.  By contrast, the question ``does there exist an \textit{obviously} strategy-proof mechanism for choice rule $c$?" is a question about the entire class of extensive-form game trees.  Recent work has shown that it is without loss of generality to restrict attention to smaller classes, such as game trees with perfect information about the history of play \citep{ashlagi_stable_2018,bade_gibbard-satterthwaite_2017,pycia_theory_2019,mackenzie2020revelation}.  However, even the smallest such class is still combinatorially complex, which makes it difficult to tell in general whether an obviously strategy-proof mechanism exists.\footnote{It is possible to look at yet smaller classes if we make more assumptions about utility functions or choice rules.  Such characterizations are known for a variety of settings, such as binary allocation problems \citep{li_obviously_2017}, single-peaked social preferences \citep{bade_gibbard-satterthwaite_2017, arribillaga_obvious_2020}, and `rich' preferences \citep{pycia_theory_2019}, as well as for top trading cycles rules \citep{troyan_obviously_2019, mandal2020obviously} and sequential allotment rules \citep{arribillaga2019all}.}

Extensive-form mechanisms can be simpler for the agents, but are they much more complicated for the designer?  To quantify this trade-off, we study the problem of determining whether there exists an obviously strategy-proof mechanism for a given choice rule, from the perspective of computational complexity.

Consider the following decision problem:  We are given agents with finite type sets, some set of outcomes, arbitrary utility functions, and a choice rule $c$.  These are represented as a table, which specifies the payoff of each type of each agent, at each  profile of reported types.  Does there exist an obviously strategy-proof (OSP) mechanism for $c$?  (That is, is $c$ OSP-implementable?)  If the answer is ``yes", then we can prove this by exhibiting such a mechanism.  Moreover, a result of \cite{bade_gibbard-satterthwaite_2017} implies that if any such mechanism exists, then there exists a mechanism small enough to be verified as OSP in time polynomial in the table size.  Thus, the decision problem lies in the complexity class \textsf{NP}.

Can we decide quickly whether $c$ is OSP-implementable, and if so construct an OSP implementation?  For a given choice rule, there can be exponentially many game trees, so a brute force search requires exponential time.

We prove that the decision problem can be solved in polynomial time.  Moreover, we show that OSP mechanisms, when they exist, can be constructed in polynomial time.

Along the way, we discover a characterization of OSP-implementable choice rules that may be of independent interest.  Given any choice rule, we define a corresponding \textit{obvious directed acyclic graph} (O-dag).  The root of the O-dag corresponds to the set of all type profiles.  Each vertex of the O-dag corresponds to a set of type profiles with a product structure.

The directed edges of the O-dag correspond to inferences that we can make by asking `obvious' questions.  That is, for any two vertices $v$ and $v'$, there is a directed edge from $v$ to $v'$ if we can make an `obvious' binary query to one agent about his type, such that the type profile lies in $v'$ if and only if the type profile lies in $v$ and the agent's answer is ``yes".   A binary query is obvious if, given the choice rule and given that the type profile is in $v$, every possible outcome from lying is no better than every possible outcome from answering truthfully. 

We show that a choice rule is OSP-implementable if and only if its O-dag is \textit{nicely connected}, meaning that every singleton vertex can be reached from the root.  Hence, one way to solve the OSP decision problem is to construct the O-dag and check whether it is nicely connected.  However, a brute force approach to this verification does not result in a polynomial-time algorithm, because the O-dag has a vertex set that grows exponentially with the type space.

Instead, our algorithm detects whether the O-dag is nicely connected without explicitly constructing the entire O-dag.  We define a greedy algorithm that starts from the root and searches for a path to each singleton vertex.  Since the algorithm is greedy, it quickly either finds such a path or returns a failure.  We then exploit the structure of the O-dag to prove that failure to greedily find a path implies that no path exists.

The greedy algorithm terminates in sub-quadratic time with respect to the input table size.  For comparison, determining whether a choice rule is strategy-proof requires running time linear in the table size.

To summarize, deciding whether there exists an OSP mechanism for a choice rule is equivalent to deciding whether the corresponding O-dag is nicely connected, which can be done in polynomial time by a greedy algorithm.

The greedy algorithm can also be used to construct OSP mechanisms in polynomial time when they exist. Specifically, the O-dag edges visited by the algorithm assemble to form the desired mechanism.

In our results described above, inputs were expressed as a table.  However, in many natural settings, choice rules and utility functions have special structure that permits a more concise description.  When we describe the same situation in a more powerful language, the problem becomes more difficult as a function of input size.

Consequently, we consider the problem of determining whether there exists an appropriate (SP/OSP) mechanism when the input utility functions and choice rules are represented by Boolean circuits.  We leverage the O-dag structure to prove that the SP problem and the OSP problem are both hard, but they are equally hard, in the sense that both problems are co-\textsf{NP}-complete.  Hence, the OSP problem is computationally comparable to the SP problem, because either problem can be reduced to the other in polynomial time.  Additionally, this result implies that when an appropriate mechanism does not exist, there is a proof of non-existence that can be verified quickly.

\section{Model and primitives}

Consider a finite set of agents $N$ with type sets $T_1,\dots,T_n$, and an outcome set $X$.  We assume private values: each agent has a utility function $u_i: T_i \times X \rightarrow \bR$.  For each type $t_i$, let the binary relation $\wpref_{t_i}$ denote the weak preference ordering over outcomes induced by $t_i$, that is, $x \wpref_{t_i} x'$ if $u_i(t_i,x) \geq u_i(t_i,x')$.  Let $T \equiv T_1\times\cdots\times T_n$, with representative element $t$.

A choice rule is a function $\choice:T \rightarrow X$.  As usual, we use $T_{-i}$ to denote $\prod_{j \in N \setminus \{i\}} T_j$, and $t_{-i}$ to denote an element of $T_{-i}$.  We use analogous notation for products of subsets of the type sets. 

\begin{definition}\label{def:SP}
  Choice rule $\choice$ is \textbf{strategy-proof} (SP) if for all $t_i$, $t'_i$, and $t_{-i}$,
  \begin{equation}
    \choice(t_i, t_{-i}) \wpref_{t_i} \choice(t'_i, t_{-i}).
  \end{equation}
\end{definition}

\subsection{Definition of extensive-form mechanisms}

Informally, a mechanism is an extensive game form for which each terminal history is associated with some outcome.  In principle, such a mechanism could have imperfect information\footnote{That is, some agent may act without knowing other agents' past actions.} or even imperfect recall.  Furthermore, agents could in principle play mixed strategies or behavioral strategies.  However,  for our purposes it is without loss of generality to restrict attention to mechanisms in which agents take turns to make public announcements about their private information, and play pure strategies \citep{mackenzie2020revelation}.\footnote{In particular, if there exists an obviously strategy-proof mechanism for a choice rule, then there exists an obviously strategy-proof mechanism in this smaller class.}  To ease notation, we restrict attention to this class, which was identified by \cite{ashlagi_stable_2018}, and refined by \cite{pycia_theory_2019}, \cite{bade_gibbard-satterthwaite_2017}, and \cite{mackenzie2020revelation}. We additionally restrict attention to mechanisms in which every type profile results in a finite play.\footnote{When dealing with finite type sets, this requirement is without loss of generality.  In principle, obvious strategy-proofness is well-defined for the larger class of meet-semilattice trees \citep{mackenzie2020revelation}, but the restriction to finite plays makes the definitions in the sequel more transparent.}

Formally, given a choice rule $c$, a \textbf{mechanism for $c$} specifies a directed rooted tree with vertex set $V$ and edge set $E$.  Essentially, we start play from the root vertex $\overline{v}$, and at each step one agent makes a declaration about his type.  We use $\out(v)$ to denote the out-edges of vertex $v$, and $Z$ to denote the terminal vertices, \textit{i.e.} those with no out-edges.  For each non-terminal vertex $v$, some agent is called to play, whom we denote $P(v)$.

Each vertex $v \in V$ is associated with a set of types $T_i^v \subseteq T_i$ for each agent $i$, and we denote $T^v = \prod_i T_i^v$.  We interpret the out-edges as reports that could be made by $P(v)$.  Hence, for any non-terminal vertex $v$, each edge $e \in \out(v)$ corresponds to a non-empty subset of $P(v)$'s types, denoted $T_{P(v)}^e$.  We require that the family of sets $(T_{P(v)}^e)_{e \in \out(v)}$ is a partition of $T_{P(v)}^v$.

We interpret $T_i^v$ as the types of $i$ consistent with the reports so far, which requires some record-keeping.  Hence we require that at the root vertex all types are possible, $T^{\overline{v}} = T$.  We also require that if $v'$ is child of $v$ reached by edge $e$, then $T_{P(v)}^{v'} = T_{P(v)}^e$ and for all $i \neq P(v)$, $T_i^{v'} = T_i^v$.

We require that every type profile results in a finite play; that is, for all $t \in T$, there exists a terminal vertex $v \in Z$ such that $t \in T^v$.  This implies that the terminal vertices $Z$ induce a partition of $T$, which is $(T^v)_{v \in Z}$.

Finally, we require that the reports yield enough information to enact the choice rule.  Formally, this means that the choice rule $c$ is measurable with respect to the partition induced by the terminal vertices, \textit{i.e.} for all $v \in Z$ and all $t, t' \in T^v$, $c(t) = c(t')$.

\subsection{Obvious strategy-proofness}

Given some non-terminal vertex $v$, we say that $t_{P(v)}$ and $t'_{P(v)}$ \textbf{diverge at $v$} if $t_{P(v)}, t'_{P(v)} \in T_{P(v)}^v$ and they make distinct reports, that is, if there is no $e \in \out(v)$ such that $t_{P(v)}, t'_{P(v)} \in T_{P(v)}^e$.

\begin{definition}\label{def:OSP}
  A mechanism for $c$ is \textbf{obviously strategy-proof} (OSP) if for any agent $i$, any vertex $v$ with $P(v) = i$, any $t_i$ and $t'_i$ that diverge at $v$, and any $t_{-i}, t'_{-i} \in T_{-i}^v$, we have that
  \begin{equation}
    \choice(t_i, t_{-i}) \wpref_{t_i} \choice(t'_i, t'_{-i}).
  \end{equation}
  We say that $c$ is \textbf{OSP-implementable} if there exists an OSP mechanism for $c$.
\end{definition}
It is worth contrasting \Cref{def:SP} and \Cref{def:OSP}.  \Cref{def:SP} is a state-by-state comparison; for any $t_{-i}$, type $t_i$ must prefer truthful reporting to any deviation. By contrast, \Cref{def:OSP} compares every possible outcome following the truthful report to every possible outcome following a deviation, but restricts this comparison to vertices at which play diverges.  This property depends on the extensive form of the mechanism.  

\Cref{def:OSP} is a stronger requirement than \Cref{def:SP}. For any choice rule $c$, if $c$ is OSP-implementable, then $c$ is SP.

The original definition of OSP restricts attention to mechanisms with a uniform bound on the depth of the game tree \citep{li_obviously_2017}.  Here we have only made the weaker requirement that every type profile results in a finite play.  \Cref{def:OSP} is equivalent to the original definition on game trees with a uniform bound.\footnote{Without a uniform bound, the original definition runs into problems with deviations that do not mimic a coherent type and prolong the game forever.  One other subtlety is that \Cref{def:OSP} considers only deviations that consistently mimic some type $t'_i$, whereas the definition of \cite{li_obviously_2017} permits deviations that mimic $t'_i$ at some vertices and mimic a distinct type $t''_i$ at other vertices.  But the existence of a `profitable' deviation of that kind implies the existence of a `profitable' deviation of the simpler kind considered by \Cref{def:OSP}, so the definitions are equivalent.} Were we to make this additional restriction, we would need to strengthen the connectedness requirement (\Cref{def:nicely_connected}), but it would not otherwise alter any of the results.

\section{A characterization of OSP-implementable choice rules}\label{sec:char_OSP}

For any given choice rule $c$, the number of mechanisms for $c$ is exponential in $|T_i|$ for each agent $i$.  In this section, we develop a machinery that associates each choice rule with a single directed acyclic graph.  We then prove that there exists an OSP mechanism for $c$ if and only if the graph is connected in a particular way.

\begin{definition}
  \label{def:odag}
  The \textbf{obvious directed acyclic graph (O-dag)} of choice rule $\choice$ is the directed acyclic graph with vertex set
  \begin{equation}
    \left\{\prod_{i} T_i' \ \middle| \ \forall i: \emptyset\subset T_i'\subseteq T_i \right\}
  \end{equation}
  and with edge set containing all edges of the form
  \begin{equation}
    T' \Rightarrow (A_i, T'_{-i})
  \end{equation}
  where $\emptyset\subset A_i\subset T_i'$, such that for all $a_i \in A_i$, all $a'_i \in T_i'\setminus A_i$, and all $t_{-i}, t'_{-i} \in T'_{-i}$
  \begin{equation}
    \choice(a_i, t_{-i}) \wpref_{a_i} \choice(a'_i, t'_{-i}) \text{ and }
    \choice(a'_i, t'_{-i}) \wpref_{a'_i} \choice(a_i, t_{-i}).
  \end{equation}
\end{definition}

We now establish some properties of the O-dag that aid in the characterization.

\begin{lemma}
  \label{lem:indedge}
  If the O-dag of $\choice$ contains an edge $T' \Rightarrow (A_i, T'_{-i})$, then for any vertex $B \subseteq T'$ such that $\emptyset\subset (A_i\cap B_i)\subset B_i$, the O-dag of $\choice$ also contains the edge $B \Rightarrow (A_i \cap B_i, B_{-i})$.
\end{lemma}

\begin{proof}
  Take any $a_i \in A_i \cap B_i$, any $a'_i \in B_i \setminus (A_i \cap B_i)$, and any $t_{-i}, t'_{-i} \in B_{-i}$.  Observe that $A_i \cap B_i \subseteq A_i$, $B_i \setminus (A_i \cap B_i) \subseteq T'_i \setminus A_i$, and $B_{-i} \subseteq T'_{-i}$.  Since the O-dag of $\choice$ contains $T' \Rightarrow (A_i, T'_{-i})$, we have that
  \begin{equation}
    \choice(a_i, t_{-i}) \wpref_{a_i} \choice(a'_i, t'_{-i}) \text{ and }
    \choice(a'_i, t'_{-i}) \wpref_{a'_i} \choice(a_i, t_{-i}).
  \end{equation}
  Hence, the O-dag of $\choice$ contains the edge $B \Rightarrow (A_i \cap B_i, B_{-i})$.
\end{proof}

To ease notation, note that the Cartesian product is distributive over intersections, so for any two sets $A = \prod_i A_i$ and $B = \prod_i B_i$ we have $A \cap B = \prod_i \left(A_i \cap B_i \right)$.

\begin{lemma}
  \label{lem:indpath}
  If the O-dag of $\choice$ contains a path from vertex $T'$ to vertex $T''$, then for any vertex $B$ such that $B_i \cap T''_i \neq \emptyset$ for all $i$, the O-dag contains a path from $B \cap T'$ to $B \cap T''$.
\end{lemma}
\begin{proof}
  Let $T^1, T^2, \ldots, T^K$ be a path with $T^1 = T'$ and $T^K = T''$.  By construction, for all $k < K$ we have an edge $T^k \Rightarrow T^{k+1}$ and for some agent $i$, $T^{k+1} = (T_i^{k+1}, T_{-i}^k)$ and $\emptyset \subset T_i^{k+1} \subset T_i^k$.
  
  Suppose $B_i \cap T_i^k \supset B_i \cap T_i^{k+1}$. Observe that $B_i \cap T_i^{k+1} \supseteq B_i \cap T''_i \neq \emptyset$.  Then by Lemma \ref{lem:indedge}, the O-dag contains the edge $B \cap T^k \Rightarrow B \cap T^{k+1}$.
  
  Suppose $B_i \cap T_i^k = B_i \cap T_i^{k+1}$.  Then, by $T_{-i}^k = T_{-i}^{k+1}$, we have that $B \cap T^k = B \cap T^{k+1}$.
  
  Hence, for all $k < K$, either there exists an edge $B \cap T^k \Rightarrow B \cap T^{k+1}$ or $B \cap T^k = B \cap T^{k+1}$, so there exists a path from $B \cap T'$ to $B \cap T''$.
\end{proof}

We now introduce notation for the `answers' that are feasible from a vertex of the O-dag.  For any vertex $T'$ in the O-dag of $c$, let $\ans_i(T')$ denote the family of sets
\begin{equation}
  \left\{A_i \ \middle| \ \text{the O-dag contains $T' \Rightarrow (A_i,T'_{-i}))$} \right\}.
\end{equation}

The symmetry of \Cref{def:odag} implies that $\ans_i(T')$ is closed under complements (relative to $T'_i$).  We now prove that it is closed under intersections and unions.

\begin{lemma}
  \label{lem:edgestructure}
  Let $T'$ be any vertex of the O-dag of $\choice$. For any agent $i$, consider some indexed sub-family $\{ A_i^k \mid k \in K \} \subseteq \ans_i(T')$ with at least one member.
  \begin{enumerate}
  \item If $\emptyset \subset \bigcap_k A_i^k$, then $\bigcap_k A_i^k \in \ans_i(T')$.
  \item If $\bigcup_k A_i^k \subset T'_i$, then $\bigcup_k A_i^k \in \ans_i(T')$.
  \end{enumerate}
\end{lemma}

\begin{proof}
  Suppose the antecedent of Clause 1, $\emptyset \subset \bigcap_k A_i^k$.  By \Cref{def:odag}, we have $\bigcap_k A_i^k \subset T'_i$.  Take any $t_i \in \bigcap_k A_i^k$.  For all $k$, $A_i^k \in \ans_i(T')$.  Thus, for all $t'_i \in \bigcup_k (T'_i \setminus  A_i^k)$, we have 
  \begin{equation}
    \forall t_{-i}, t'_{-i} \in T'_{-i}: \choice(t_i,t_{-i}) \wpref_{t_i} \choice(t'_i,t'_{-i})
  \end{equation}
  \begin{equation}
    \forall t_{-i}, t'_{-i} \in T'_{-i}: \choice(t'_i,t'_{-i}) \wpref_{t'_i}  \choice(t_i,t_{-i}).
  \end{equation}
  Moreover, $\bigcup_k ( T'_i \setminus A_i^k ) = T'_i \setminus \bigcap_k A_i^k $, so $\bigcap_k A_i^k \in \ans_i(T')$.  This proves Clause 1 of \Cref{lem:edgestructure}.
  
  Suppose the antecedent of Clause 2, $\bigcup_k A_i^k \subset T'_i$.  Note that $T'_i \setminus \bigcup_k A_i^k = \bigcap_k ( T'_i \setminus A_i^k )$, so we have
  \begin{equation}\label{eq:edgestructure_1}
    \emptyset \subset \bigcap_k (T'_i \setminus A_i^k).
  \end{equation}
  For all $k$, $A_i^k \in \ans_i(T')$, so by closure under relative complement we have that for all $k$,
  \begin{equation}\label{eq:edgestructure_2}
    T'_i \setminus A_i^k \in \ans_i(T').
  \end{equation}
  By Clause 1, \eqref{eq:edgestructure_1}, and \eqref{eq:edgestructure_2}, we have that
  \begin{equation}\label{eq:edgestructure_3}
    \bigcap_k (T'_i \setminus A_i^k) \in \ans_i(T').
  \end{equation}    
  Substituting $\bigcap_k ( T'_i \setminus A_i^k ) = T'_i \setminus \bigcup_k A_i^k$ in \eqref{eq:edgestructure_3} and applying closure under relative complement yields $\bigcup_k A_i^k \in \ans_i(T')$,  which proves Clause 2 of \Cref{lem:edgestructure}.
\end{proof}

These observations imply there exists a partition of $T_i'$ that contains all minimal elements of $\ans_i(T')$, which we now state formally.

\begin{lemma}\label{lem:min_partition}
  Let $T'$ be any vertex in the O-dag of $\choice$.  For any $i$, if $\ans_i(T')$ is non-empty, then there exists a partition $\minans_i(T')$ of $T'_i$ such that:
  \begin{enumerate}
  \item $\minans_i(T') \subseteq \ans_i(T')$
  \item For any $A_i \in \ans_i(T')$ and any $t_i \in A_i$, the unique element $A'_i \in \minans_i(T')$ such that $t_i \in A'_i$ satisfies $A'_i \subseteq A_i$.
  \end{enumerate}
\end{lemma}
\begin{proof}
  For any $t_i \in T'_i$, let us define
  \begin{equation}
    \ans_i(t_i,T') \equiv \left\{A_i \in \ans_i(T') \ \middle| \ t_i \in A_i \right\}.
  \end{equation}
  Note that by $\ans_i(T')$ non-empty, there exists at least one element $A'_i \in \ans_i(T')$.  By closure under relative complement, we have $T'_i \setminus A'_i \in \ans_i(T')$.  Hence, for all $t_i \in T'_i$, $\ans_i(t_i,T')$ is non-empty. Now we define
  \begin{equation}
    \minans_i(T') \equiv \left\{ \bigcap \ans_i(t_i,T')  \ \middle| \ t_i \in T'_i \right\}.
  \end{equation}
  Observe that for all $t_i \in T'_i$, we have $t_i \in \bigcap \ans_i(t_i,T')$, which implies that $\bigcup \minans_i(T') = T'$. 
  
  Next we prove Clause 1 of \Cref{lem:min_partition}.  Take any $t_i \in T'_i$. By construction $\emptyset \subset \{t_i\} \subseteq \bigcap \ans_i(t_i,T')$. By \Cref{lem:edgestructure}, $\bigcap \ans_i(t_i,T') \in \ans_i(T')$.  And $t_i$ was chosen arbitrarily, so we have that $\minans_i(T') \subseteq \ans_i(T')$, which proves Clause 1.

  Now we prove that any two elements of $\minans_i(T')$ are either disjoint or identical. Take any $A_i, A'_i \in \minans_i(T')$.  By Clause 1, $A_i, A'_i \in  \ans_i(T')$.  By construction, there exists $t_i$ such that $A_i = \bigcap \ans_i(t_i,T')$.  Either $t_i \in A'_i$ or $t_i \in T'_i \setminus A'_i$.

  Suppose $t_i \in A'_i$.  By \Cref{lem:edgestructure}, $A_i \cap A'_i \in \ans_i(T')$, so  $A_i \cap A'_i \in \ans_i(t_i,T')$, which yields $A_i = \bigcap \ans_i(t_i,T') \subseteq A_i \cap A'_i \subseteq A'_i$.

  Suppose $t_i \in T'_i \setminus A'_i$.  By closure under relative complement,  $T'_i \setminus A'_i \in \ans_i(T')$. By \Cref{lem:edgestructure}, $A_i \cap (T'_i \setminus A'_i) \in \ans_i(T')$, so $A_i \cap (T'_i \setminus A'_i) \in \ans_i(t_i,T')$.  Thus, we have $A_i = \bigcap \ans_i(t_i,T') \subseteq A_i \cap (T'_i \setminus A'_i) \subseteq T'_i \setminus A'_i$.  This implies that $A_i \cap A'_i = \emptyset$.

  We have proved that either $A_i \cap A'_i = \emptyset$ or $A_i \subseteq A'_i$.  By symmetry, we also have that $A_i \cap A'_i = \emptyset$ or $A_i \supseteq A'_i$. Thus, $A_i$ and $A'_i$ are either disjoint or identical.   We have established that $\minans_i(T')$ is a partition of $T'$.

  Clause 2 follows by construction of $\minans_i(T')$.  Observe that for any $A_i \in \ans_i(T')$ and any $t_i \in A_i$, we have that $A_i \in \ans_i(t_i,T')$.  So for $A'_i = \bigcap \ans_i(t_i,T')$, we have that $A'_i \subseteq A_i$, which completes the proof.
\end{proof}

We now define a connectedness property for O-dags.

\begin{definition}
  \label{def:nicely_connected}
  The O-dag of $\choice$ is \textbf{nicely connected} if for every $t \in T$, there exists a path from the root $T$ to the vertex $\{t\}$.
\end{definition}

The OSP-implementable choice rules are characterized by this connectedness property.  This enables us to reduce a question about the class of all mechanisms for $\choice$ to a question about the O-dag of $\choice$.

\begin{theorem}\label{thm:ODAG_OSP}
  For any choice rule $\choice$, there exists an OSP mechanism for $\choice$ if and only if the O-dag of $\choice$ is nicely connected.
\end{theorem}

\begin{proof}
  Suppose there exists an OSP mechanism for $\choice$, and denote its vertex set $V$ and its edge set $E$.  Take any $t \in T$.  There exists a terminal vertex $z$ such that $t \in T^z$.  Hence there exists a path $v^1, v^2, \ldots, v^K$ in the tree such that $\overline{v} = v^1$ and $v^K = z$.

  We now prove that for all $k < K$, the O-dag of $\choice$ contains the edge $T^{v^k} \Rightarrow T^{v^{k+1}}$.  Let $i = P(v^k)$.  By construction, $\emptyset \subset T_{i}^{v^{k+1}} \subset T_{i}^{v^k}$ and $T_{-i}^{v^{k+1}} = T_{-i}^{v^k}$. Take any $t_i \in T_i^{v^{k+1}}$ and any $t'_i \in T_i^{v^k} \setminus T_i^{v^{k+1}}$.  Since $t_i$ and $t'_i$ diverge at $v^k$ and the mechanism is OSP, it follows that for any $t_{-i},t'_{-i} \in T_{-i}^{v^k}$,
  \begin{equation}
    \choice(t_i,t_{-i}) \wpref_{t_i} \choice(t'_i,t'_{-i}).
  \end{equation}
  Symmetry of \Cref{def:OSP} implies that
  \begin{equation}
    \choice(t'_i,t'_{-i}) \wpref_{t'_i} \choice(t_i,t_{-i}).
  \end{equation}
  which proves that the O-dag of $\choice$ contains the edge $T^{v^k} \Rightarrow T^{v^{k+1}}$.

  Since $\choice$ is measurable with respect to the terminal vertices of the mechanism, we have that for all $t', t'' \in T^{v^K}$, $\choice(t') = \choice(t'')$.  This implies that if $T^{v^K} \neq \{t\}$, then the O-dag contains the edge $T^{v^K} \Rightarrow \{t\}$.  Thus, for all $t$, the O-dag of $\choice$ has a path from the root $T$ to $\{t\}$.  This proves that if there exists an OSP mechanism for $\choice$, then the O-dag of $\choice$ is nicely connected.  

  Suppose that the O-dag of $\choice$ is nicely connected.  We now construct an OSP mechanism for $\choice$.  Essentially, we cycle through the bidders in a fixed order, keeping track of the type profiles consistent with all previous reports, $\tilde{T}^k$.  At the $k$th step, we ask bidder $i$ to report a cell of the minimal partition $\minans_i(\tilde{T}^{k})$ if it is non-empty.  Otherwise, we skip bidder $i$.  The key is to show that the O-dag being nicely connected implies that this procedure does not get stuck for any type profile.

  Formally, we generate the mechanism using the following procedure:  As before, we identify agents with integers $1, \ldots, n$.  Let $i^k \equiv (k\mod n) + 1$.  We initialize $\tilde{T}^0 = T$.  For $k = 0, 1, 2, \ldots$, if $\ans_{i^k}(\tilde{T}^k)$ is non-empty, then $i^k$ reports $A_{i^k} \in \minans_{i^k}(\tilde{T}^k)$ such that $t_{i^k} \in A_{i^k}$, which is well-defined and unique by \Cref{lem:min_partition}.  We then define $\tilde{T}^{k+1} \equiv (A_{i^k}, \tilde{T}_{-i^k}^k)$.  If $\ans_{i^k}(\tilde{T}^k)$ is empty, we define $\tilde{T}^{k+1} \equiv \tilde{T}^{k}$.  The procedure terminates if $\tilde{T}^{k}$ is singleton.

  We now prove that the above procedure generates a mechanism for $c$.  In particular, we show that for any $t \in T$, there is a terminal vertex $v$ of the mechanism such that $T^v = \{t\}$, which implies that every type profile results in a finite play, and also that the choice rule is measurable with respect to the partition induced by the terminal vertices.

  Take any $t \in T$.  Since the O-dag is nicely connected, there exists a path in the O-dag from $T$ to $\{t\}$, which we denote $T^0, T^1, \ldots, T^K$.  Consider the sequence generated by the procedure when the type profile is $t$, which we denote $(\tilde{T}_k)_{k = 1}^\infty$.  For convenience, if the procedure terminates at step $\tilde{K}$, we define $\tilde{T}^k = \tilde{T}^{\tilde{K}}$ for $k > \tilde{K}$.

  We prove by induction that $\tilde{T}^{nk} \subseteq T^k$ for all $k \leq K$.  The inductive hypothesis holds for $k = 0$, since $\tilde{T}^0 = T = T^0$.  Suppose it holds for some $k < K$; we now prove it holds for $k + 1$.  Since $T^0, T^1, \ldots, T^K$ is a path in the O-dag, there exists some agent $j$ and some $A_j \subset T_j^k$ such that $(A_j, T_{-j}^k) = T^{k+1}$, and the O-dag contains edge
  \begin{equation}\label{eq:use_lemma_indedge_0}
    T_j^k \Rightarrow (A_j, T_{-j}^k)
  \end{equation}
  If $\tilde{T}_j^{nk} \subseteq A_j$, then $\tilde{T}^{n(k+1)} \subseteq \tilde{T}^{nk} \subseteq (A_j, T_{-j}^k) = T^{k+1}$ and we are done.  Suppose not, so $A_j \cap \tilde{T}_j^{nk} \subset \tilde{T}_j^{nk}$.  Observe that $t \in A_j \cap \tilde{T}_j^{nk}$, so we have that
  \begin{equation}\label{eq:use_lemma_indedge_1}
    \emptyset \subset A_j \cap \tilde{T}_j^{nk} \subset \tilde{T}_j^{nk}.
  \end{equation}
  $i^l = j$ for some $l$ weakly between $nk$ and $n(k + 1) - 1$, and $\tilde{T}_j^l = \tilde{T}_j^{nk}$.  By construction and then by the inductive hypothesis,
  \begin{equation}\label{eq:use_lemma_indedge_2}
    \tilde{T}^l \subseteq \tilde{T}^{nk} \subseteq T^k.
  \end{equation}
  By \eqref{eq:use_lemma_indedge_0}, \eqref{eq:use_lemma_indedge_1}, \eqref{eq:use_lemma_indedge_2}, and \Cref{lem:indedge}, the O-dag contains the edge
  \begin{equation}\label{eq:intersect_ok}
    \tilde{T}^l \Rightarrow (A_j \cap \tilde{T}_j^{l}, \tilde{T}^l_{-j}).
  \end{equation}
  Since $\minans_j(\tilde{T}^l)$ consists of the minimal elements of $\ans_j(\tilde{T}^l)$, \eqref{eq:intersect_ok} and the definition of the sequence imply that
  \begin{equation}
    \tilde{T}_j^{l+1} \subseteq A_j \cap \tilde{T}_j^{l}.
  \end{equation}
  Hence,
  \begin{equation}
    \tilde{T}^{n(k+1)} \subseteq \tilde{T}^{l+1} \subseteq ( A_j \cap \tilde{T}_j^{l}, \tilde{T}_{-j}^{nk}) \subseteq (A_j, T_{-j}^k) = T^{k + 1}
  \end{equation}
  which proves the inductive step.  Thus, for all $k \leq K$, we have that $\tilde{T}^{nk} \subseteq T^k$.  In particular, $\tilde{T}^{nK} \subseteq T^K = \{t\}$, so the procedure terminates after no more than $nK$ steps.

  Consequently, for every type profile $t$, there exists a terminal vertex $v$ in the corresponding mechanism such that $T^v = \{t\}$, so the procedure generates a mechanism for $\choice$.  

  Finally, the mechanism generated by the procedure is OSP.  In particular, if $t_i$ and $t'_i$ diverge at some vertex $v$ of the mechanism, then $t_i$ and $t'_i$ are in distinct cells of $\minans_i(T^v)$.  Hence, we have that for any $t_{-i},t'_{-i} \in T_{-i}^v$,
  \begin{equation}
    \choice(t_i,t_{-i}) \wpref_{t_i} \choice(t'_i,t'_{-i}).
  \end{equation}
  Thus, if the O-dag of $\choice$ is nicely connected, then there exists an OSP mechanism for $\choice$.
\end{proof}

We will shortly use the O-dag characterization to answer algorithmic questions about OSP mechanisms.  However, the characterization may be of independent interest.  For instance, \Cref{lem:indpath} and \Cref{thm:ODAG_OSP} together yield a tractable method to prove the non-existence of OSP mechanisms for particular choice rules, stated in the following corollary.

\begin{corollary}\label{corr:rule_out}
  For any choice rule $\choice$, if the O-dag of $\choice$ contains a non-singleton vertex $T'$, and $t \in T'$ such that there is no path from $T'$ to $\{t\}$, then $\choice$ is not OSP-implementable.
\end{corollary}

When type spaces are finite, all paths in the O-dag are finite. This property gives the stronger condition below.

\begin{corollary}\label{corr:finite_rule_out}
  If all $T_i$ are finite, then a choice rule $\choice$ is OSP-implementable if and only if every non-singleton vertex in its O-dag has a child.
\end{corollary}
\begin{proof}
  If the O-dag of $\choice$ contains a non-singleton vertex with no child, then Corollary~\ref{corr:rule_out} implies that $\choice$ is not OSP-implementable.

  For the converse, let every non-singleton vertex of the O-dag have a child. Consider any $t\in T$. Inductively construct a path $T^1,T^2,\dots$ with $T^1=T$ such that $T^{k+1}$ is any child of $T^k$ that contains $t$. The path ends upon reaching some $T^K$ with no such child. If $T^K$ were non-singleton, then it would have some children $(A_i,T_{-i}^K)$ and $(T_i^K\setminus A_i,T_{-i}^K)$, one of which contains $t$, a contradiction. Thus $T^K=\{t\}$. Therefore the O-dag is nicely connected, so $\choice$ is OSP-implementable by \Cref{thm:ODAG_OSP}.
\end{proof}

\section{An algorithm for deciding whether a choice rule is OSP-implementable}\label{sec:polytime}

This section studies the computational problem of determining whether a choice rule is OSP-implementable, assuming finite type spaces.  For an arbitrary choice rule $\choice$, we use the structure of the O-dag of $\choice$ to efficiently determine if $\choice$ is OSP-implementable.

Suppose we have finite type sets $(T_i)_{i \in N}$, outcome set $X$, a choice rule $\choice:\prod_i T_i \rightarrow X$, and utility functions $u_i: T_i \times X \rightarrow \bR$.  An instance of the decision problem is comprised by $\sum_i|T_i|$ functions, one for each type $t_i$,
\begin{equation}
  u_i(t_i,\cdot)\circ\choice:T_1\times\cdots\times T_n\rightarrow\bR.
\end{equation}
Each of these functions requires $\prod_i|T_i|$ real numbers to express in general.  In our decision problem, we take these functions as the input.  Hence, we define
\begin{equation}
  |\choice|=\left(\sum_i|T_i|\right)\prod_i|T_i|
\end{equation}
which is the size of the table that specifies the payoffs to each type of each agent, for each profile of reports.  Running times for algorithms on $\choice$ will be expressed as a function of $|\choice|$.\footnote{A limitation of this approach is that, because we take the payoffs induced by the choice rule as an input, we abstract from the difficulty of computing the choice rule itself.  For instance, there are settings in which welfare-maximizing choice rules are not tractable to compute \citep{hartline2015non, leyton2017economics}.  In such settings, even verifying strategy-proofness can be difficult.}

First consider the problem of deciding whether $\choice$ is strategy-proof. For each $i$ and each $t_i\in T_i$, it is necessary and sufficient to verify for all $t_i'\in T_i$ and all $t_{-i}\in T_{-i}$ that $u_{t_i}(\choice(t_i,t_{-i}))\geq u_{t_i}(\choice(t_i',t_{-i}))$. This verification requires $\sum_i|T_i|^2|T_{-i}|=|\choice|$ steps, so deciding whether $\choice$ is SP requires linear time.

The following result states that it is not much more difficult to determine if $\choice$ is OSP-implementable. 

\begin{theorem}
  \label{thm:subquad_decide}
  There exists a $O(|\choice|\sum_i|T_i|)$-time algorithm for deciding whether $\choice$ is OSP-implementable.
\end{theorem}

To put this result in context, consider that while the number of histories in an arbitrary OSP mechanism for $\choice$ may grow exponentially in $|\choice|$, the gradual revelation principle of~\cite{bade_gibbard-satterthwaite_2017} guarantees that every OSP mechanism for $\choice$ has an OSP implementation of depth at most $\sum_i|T_i|$ and with at most $\prod_i|T_i|$ distinct terminal histories, so that $\choice$ has an OSP implementation of size at most $|\choice|$. It can be verified in polynomial time that any such implementation is OSP, so it follows that the problem of deciding whether $\choice$ is OSP-implementable lies in \textsf{NP}. However, the number of (extensive-form) mechanisms of size at most $|\choice|$ grows exponentially in $|\choice|$, so a brute force search through all mechanisms would require exponential time. A brute force search through the O-dag of $\choice$ is similarly inefficient, as the number of vertices in the O-dag grows exponentially in $|\choice|$. It is therefore not immediately clear that the problem of deciding whether $\choice$ is OSP-implementable lies in \textsf{P}; this is the contribution of \Cref{thm:subquad_decide}.  Note that $|\choice|\sum_i|T_i| \leq |\choice|^2$, with strict inequality if $\prod_i |T_i| > 1$, so \Cref{thm:subquad_decide} asserts that the problem can be decided in sub-quadratic time.

The proof of \Cref{thm:subquad_decide} leverages the structure of the O-dag presented in \Cref{sec:char_OSP} to efficiently find paths from the root to each singleton vertex, without a brute force search through all vertices. We will use the following lemma.

\begin{lemma}
  \label{lem:findchildren}
  For any vertex $T'$ in the O-dag of $c$, there exists a $O((\sum_i|T_i'|)\prod_i|T_i'|)$-time algorithm that determines for each $i$ if $\ans_i(T')$ is nonempty, and if so then computes the partition $\minans_i(T')$.
\end{lemma}
\begin{proof}
  For each agent $i$, construct an undirected graph $G_i(T')$ with vertex set $T'_i$ and edge set containing all pairs $\{t_i',t_i''\}$ such that
  \begin{equation}
    \min_{t'_{-i} \in T'_{-i}} u_i(t'_i,\choice(t_i',t'_{-i})) < \max_{t'_{-i} \in T'_{-i}} u_i(t'_i,\choice(t_i'',t'_{-i}))
  \end{equation}
  or
  \begin{equation}
    \min_{t'_{-i} \in T'_{-i}} u_i(t''_i,\choice(t_i'',t'_{-i})) < \max_{t'_{-i} \in T'_{-i}} u_i(t''_i,\choice(t'_i,t'_{-i}))
  \end{equation}
  That is, if $G_i(T')$ has an edge $\{t_i',t_i''\}$, then any child of $T'$ in the O-dag of $\choice$ contains either none or both of $t_i',t_i''$.  Each maximization or minimization takes $O(|T_{-i}'|)$ time, so the construction of $G_i(T')$ takes $O(|T_i'|^2\cdot|T_{-i}'|)$ time.  Thus, to construct $G_i(T')$ for all $i$ takes $O((\sum_i|T_i'|)\prod_i|T_i'|)$ time.
  
  For each $i$, the connected components of $G_i(T')$ can be found in $O(|T_i|^2)$ time.\footnote{Given an undirected graph with vertex set $V$ and edge set $E$, its connected components can be found in time proportional to $\max\{|V|,|E|\} \leq |V|^2$ \citep{hopcroft1973algorithm}.}  To do this for all $i$ takes $O(\sum_i |T_i|^2)$ time. By Definition~\ref{def:odag} and the construction of $G_i(T')$, a subset $\emptyset\subset A_i\subset T_i'$ belongs to $\ans_i(T')$ iff there are no edges in $G_i(T')$ between $A_i$ and $T_i'\setminus A_i$. Thus $\ans_i(T')=\emptyset$ iff $G_i(T')$ has a single connected component, and if $\ans_i(T')\neq\emptyset$, then $\minans_i(T')$ is the partition of $T_i$ given by the connected components of $G_i(T')$. The total running time of the algorithm is $O((\sum_i|T_i'|)\prod_i|T_i'|)+O(\sum_i|T_i'|^2)=O((\sum_i|T_i'|)\prod_i|T_i'|)$.
\end{proof}

We now define a recursive algorithm that, given any vertex $T'$ of the O-dag, determines whether the sub-dag rooted at $T'$ is nicely connected.

\begin{algorithm}[H]
  \SetKwFunction{NC}{NicelyConnected}
  \SetKwProg{Fn}{Function}{}{}
  \Fn{\NC{$T'$}}{
    \If{$|T'|=1$}{
      \Return{True}
    }\ElseIf{$\ans_i(T')=\emptyset\;\forall i\in N$}{
      \Return{False}
    }\Else{
      Choose some $i=i(T')\in N$ with $\ans_i(T')\neq\emptyset$ \\
      \Return{$\bigwedge_{A_i\in\minans_i(T')}$\NC{$A_i,T_{-i}'$}}\footnote{That is, return \emph{True} iff for all $A_i\in\minans_i(T')$, \NC{$A_i,T_{-i}'$} returns \emph{True}.}
    }
  }
  \caption{\label{alg:niceconn} Determine if subdag of O-dag rooted at $T'$ is nicely connected.}
\end{algorithm}

\begin{proof}[Proof of \Cref{thm:subquad_decide}]
  By \Cref{thm:ODAG_OSP}, choice rule $\choice$ is OSP-implementable if and only if the O-dag of $\choice$ is nicely connected. Therefore Algorithm~\ref{alg:niceconn} presents the desired algorithm \NC{$T'$}, which determines if the subdag of the O-dag of $\choice$ rooted at any vertex $T'$ is nicely connected.

  We show the correctness of \NC by induction on the vertex $T'$. For the base case, if $|T'|=1$ then the subdag rooted at $T'$ is by definition nicely connected. For the inductive step, assume that $|T'|>1$, and that for any vertex $T''\subset T'$, \NC{$T''$} returns True iff the subdag rooted at $T''$ is nicely connected. We now show that \NC{$T'$} returns True iff the subdag rooted at $T'$ is nicely connected.

  First assume that \NC{$T'$} returns True, so that there exists some $i$ such that $\ans_i(T')\neq\emptyset$ and for each $A_i\in\minans_i(T')$, the subdag rooted at $(A_i,T_{-i}')$ is nicely connected. Then for any $t\in T'$, by definition $t_i\in A_i$ for some $A_i\in\minans_i(T')$, so by the inductive hypothesis the O-dag of $\choice$ must have a path from $T'$ to $\{t\}$ whose first edge is $T' \Rightarrow (A_i,T_{-i}')$. Therefore if \NC{$T'$} returns True, then the subdag rooted at $T'$ is nicely connected.

  Now assume that the subdag rooted at $T'$ is nicely connected. By definition $T'$ must have outgoing edges in the O-dag of $\choice$, so there exists some $i$ with $\ans_i(T')\neq\emptyset$. Then for any $A_i\in\minans_i(T')$ and any $t\in (A_i,T_{-i}')$, nice connectedness at $T'$ implies that the O-dag of $\choice$ contains a path from $T'$ to $\{t\}$, from which it follows by Lemma~\ref{lem:indpath} that the O-dag contains a path from $(A_i,T_{-i}')$ to $\{t\}$. Thus the subdag rooted at $(A_i,T_{-i}')$ is nicely connected for all $A_i\in\minans_i(T')$, so \NC{$T'$} returns True, completing the inductive step.

  It remains to be shown that \NC{$T$} runs in $O((\sum_i|T_i)|\choice|)$ time. For any $d\geq 1$, let $\cT_d$ denote the set of all vertices $T'$ such that \NC{$T'$} is called at recursive depth $d$ during the execution of \NC{$T$}. For example, $\cT_1=\{T\}$, and either $\cT_2=\emptyset$ or $\cT_2=\minans_i(T)$ for some $i$. Because any call to \NC{$T'$} will only recursively invoke \NC{$T''$} if $T''_i\subseteq T_i'$ for all $i$ with a proper inclusion for some $i$, it follows that the recursion can go at most $\sum_i|T_i|$ levels deep, so $\cT_d=\emptyset$ for $d>\sum_i|T_i|$.

  For any vertex $T'$ in the O-dag of $\choice$, Lemma~\ref{lem:findchildren} gives a $O((\sum_i|T_i'|)\prod_i|T_i'|)$-time algorithm for determining if $\ans_i(T')=\emptyset$ and computing $\minans_i(T')$ for all $i$. Therefore the running time of \NC{$T$} is at most
  \begin{equation}
    O\left(\sum_{d=1}^{\sum_i|T_i|}\sum_{T'\in\cT_d}\left(\sum_i|T_i'|\right)\prod_i|T_i'|\right).
  \end{equation}
  Because $\cT_1=\{T\}$, the $d=1$ term of the above sum equals $|\choice|$. Thus to show the desired running time bound of $O((\sum_i|T_i|)|\choice|)$, it is sufficient to show that for all $d\geq 1$, the $d$-term of the sum above bounds the $(d+1)$-term, that is,
  \begin{equation}
    \label{eq:runtimedepth}
    \sum_{T'\in\cT_d}\left(\sum_i|T_i'|\right)\prod_i|T_i'| \geq \sum_{T'\in\cT_{d+1}}\left(\sum_i|T_i'|\right)\prod_i|T_i'|.
  \end{equation}
  Consider any $T'\in\cT_d$. If $\ans_i(T')=\emptyset$ for all $i$, then \NC{$T'$} makes no recursive calls. Otherwise, for some $i=i(T')$ this invocation calls \NC{$A_i,T_{-i}$} for all $A_i\in\minans_i(T')$, so
  \begin{equation}
    \label{eq:recursivecalls}
    \cT_{d+1} = \bigcup_{T'\in\cT_d \mid \exists i=i(T')\text{ with }\ans_i(T')\neq\emptyset}\{(A_i,T_{-i}') \mid A_i\in\minans_{i}(T')\}.
  \end{equation}
  For any $T'$ and $i=i(T')$,
  \begin{equation}
    \left(\sum_j|T_j'|\right)\prod_j|T_j'| = \sum_{A_i\in\minans_i(T')}\left(\sum_j|T_j'|\right)|A_i|\prod_{j\neq i}|T_j'| \geq \sum_{A_i\in\minans_i(T')}\left(|A_i|+\sum_{j\neq i}|T_j'|\right)|A_i|\prod_{j\neq i}|T_j'|.
  \end{equation}
  Summing this inequality over all $T'\in\cT_d$ and applying~(\ref{eq:recursivecalls}) gives~(\ref{eq:runtimedepth}), as desired.
\end{proof}

\Cref{thm:subquad_decide} guarantees that we can decide in polynomial time whether there exists an OSP mechanism for $\choice$. The following theorem shows that we can construct such an OSP mechanism in polynomial time, if one exists.

\begin{theorem}
  If $\choice$ is OSP-implementable, then there exists $O((\sum_i|T_i|)|\choice|)$-time algorithm that constructs an OSP mechanism for $\choice$.
\end{theorem}
\begin{proof}
  Let $\cM(T')$ be the subtree of the O-dag of $\choice$ given by the recursion tree from a call to \NC{$T'$}. Formally, the vertices of $\cM(T')$ are the O-dag vertices $T''$ for which \NC{$T''$} is called at some point during the execution of \NC{$T'$}, and the children in $\cM(T')$ of $T''$ are the vertices $\{(A_{i(T'')},T_{-i(T'')}):A_{i(T'')}\in\minans_{i(T'')}(T'')\}$. The tree $\cM(T)$ produced by \NC{$T$} is by definition an OSP mechanism for $c$, where the player at vertex $T'$ is $i(T')$. The function \NC{$T'$} may be modified to store and return the tree $\cM(T')$ during its execution with only a constant factor slowdown in running time, so the result follows by the running time analysis of \NC in \Cref{thm:subquad_decide}.
\end{proof}

\section{Specifying choice rules with circuits}\label{sec:circuits}

So far, we have studied decision problems for which the inputs are tables, specifying the payoffs to each type of each agent, at each profile of reports.  However, in economic applications, utility functions and choice rules often have special structure that permits a more concise description.

In this section, we investigate what happens when we make the decision problem harder by using a more powerful input language.  In particular, we modify the computational model to represent utility functions and choice rules as Boolean circuits.\footnote{Our results in this section also hold for less powerful computational models such as Boolean formulae.}

We again consider agents $N$ with finite type sets $T_1,\dots,T_n$ and outcome set $X$. Each agent $i$ has utility function $u_i:T_i\times X\rightarrow\{0,\dots,|X|-1\}$, which reflects arbitrary ordinal preferences\footnote{Since our concern here is with deterministic strategy-proof choice rules, it is without loss of generality to limit attention to ordinal preferences.} over outcomes in $X$.

Utility functions $u_i$ and choice rule $\choice$ are represented as Boolean circuits $C_{u_i}$ and $C_{\choice}$ respectively. Letting each $T_i=\{0,\dots,|T_i|-1\}$, the information $(N,(T_1,\dots,T_n),(u_1,\dots,u_n),\choice)$ can be encoded as follows:
\begin{equation}
  \label{eq:encoding}
  (1^{|T_1|},\dots,1^{|T_n|},C_{u_1},\dots,C_{u_n},C_{\choice}).
\end{equation}
The outcome set $X$ is implicitly defined by the output bits of $C_{\choice}$. Note that the sizes $|T_i|$ are expressed in unary, using $\Theta(|T_i|)$ bits as opposed to $\Theta(\log|T_i|)$ bits, which permits a polynomial time algorithm with the above input to use $\text{poly}(\sum_i|T_i|)$ time.\footnote{If instead choice rules were encoded as $(|T_1|,\dots,|T_n|,C_{u_1},\dots,C_{u_n},C_{\choice})$, then \Cref{thm:versp} still holds, but the proof of \Cref{thm:verosp} breaks, as the verification algorithm would no longer run in polynomial time.}

The following results show that when choice rules are expressed as in~(\ref{eq:encoding}), deciding whether a choice rule is SP is co-\textsf{NP}-complete, and so is deciding whether a choice rule is OSP-implementable.

\begin{theorem}
  \label{thm:versp}
  Under the circuit model, deciding whether a choice rule is SP is co-\textsf{NP}-complete.
\end{theorem}
\begin{proof}
  To see that deciding SP-implementability is in co-\textsf{NP}, consider that $\choice$ is not SP iff there exists some $t_i,t_i'\in T_i,t_{-i}\in T_{-i}$ such that
  \begin{equation}
    u_i(t_i,\choice(t_i,t_{-i}))<u_i(t_i,\choice(t_i',t_{-i})).
  \end{equation}
  Thus a proof that $\choice$ is not SP may consist of $(i,t_i,t_i',t_{-i})$, and the verification of the above inequality requires time $O(|C_{u_i}|+|C_{\choice}|)$.

  We use a reduction from Boolean satisfiability (SAT) to show that deciding whether a choice rule is SP is co-\textsf{NP}-complete. Fix a Boolean formula $C_{\text{form}}$, which is by definition also a circuit. Let $n$ be one plus the number of inputs to $C_{\text{form}}$. Let $T_i=\{t_i^0,t_i^1\}$ for all $i$ and $X=\{0,1\}$. Fix some agent $i$, and define $u_i(t_i^0,x)=x$ and $u_i(t_i^1,x)=0$ for $x\in X$. For all $j\in N\setminus\{i\}$, $t_j\in T_j$, and $x\in X$, let $u_j(t_j,x)=0$. For $t_{-i}\in T_{-i}$, let $c(t_i^0,t_{-i})=0$ and $c(t_i^1,t_{-i})=C_{\text{form}}(t_{-i})$, where $C_{\text{form}}(t_{-i})$ denotes the evaluation of $C_{\text{form}}$ on input $t_{-i}\in T_{-i}\cong\{0,1\}^{n-1}$.

  The construction of $(1^{|T_1|},\dots,1^{|T_n|},C_{u_1},\dots,C_{u_n},C_{\choice})$ requires time $O(C_{\text{form}})$, as each $u_i$ has has two input bits and one output bit and therefore may be constructed in constant time, while $C_{\choice}$ can be constructed in time $O(C_{\text{form}})$. Specifically, $C_{\choice}$ is given by conditioning the output of $C_{\text{form}}(t_{-i})$ on the value of $t_i$, which requires adding a constant number of gates to $C_{\text{form}}$.

  By definition, $\choice$ is not SP iff agent $i$ with type $t_i^0$ may ever benefit from a deviation, that is, if there exists some $t_{-i}\in T_{-i}$ for which $c(t_i^1,t_{-i})=1$. This equality holds iff $C_{\text{form}}(t_{-i})=1$, so $\choice$ is not SP iff $C_{\text{form}}$ is satisfiable. Thus deciding whether a choice rule is SP-implementable is co-\textsf{NP}-complete because SAT is \textsf{NP}-complete.
\end{proof}

Next, we apply the O-dag structure to show that the OSP-implementability decision problem is also co-\textsf{NP}-complete. The main idea in the proof is to use a non-singleton vertex with no child in the O-dag as a witness to non-OSP-implementability. This approach is justified by Corollary~\ref{corr:finite_rule_out}.

\begin{theorem}
  \label{thm:verosp}
  Under the circuit model, deciding whether a choice rule is OSP-implementable is co-\textsf{NP}-complete.
\end{theorem}

To place \Cref{thm:verosp} in context, first observe that without the results in this paper, basic reasoning only places the OSP-implementability decision problem under the circuit model in the larger complexity class \textsf{NEXPTIME}. Specifically, the number of terminal vertices in an extensive-form mechanism for $\choice$ can equal $\prod_i|T_i|$, which can be exponentially large in the input size, for instance if all type sets $T_i$ and circuits $C_{u_i}$ and $C_{\choice}$ have constant size as $n$ grows. It then requires exponential time in the size of the circuit encoding~(\ref{eq:encoding}) to verify if a given extensive-form mechanism for $\choice$ is OSP, so the OSP-implementability decision problem lies in \textsf{NEXPTIME}.

Our results from \Cref{sec:polytime} imply that deciding OSP-implementability under the circuit model in fact lies in the smaller class $\textsf{PSPACE}\subseteq\textsf{NEXPTIME}$, as \Cref{alg:niceconn} for verifying whether a choice rule is OSP-implementable requires space growing polynomially in the size of the circuit encoding~(\ref{eq:encoding}).

\Cref{thm:verosp} shows the stronger\footnote{It is well known that $\text{co-\textsf{NP}}\subseteq\textsf{PSPACE}$, and it is widely conjectured that this inclusion is strict.} result that deciding OSP-implementability under the circuit model is co-\textsf{NP}-complete, thereby completely resolving the problem's complexity. Note that in contrast to the reasoning above for \textsf{NEXPTIME}, which gave an exponential-time algorithm for verifying that a given implementation of $\choice$ is OSP, we prove \Cref{thm:verosp} by giving a polynomial-time algorithm for verifying that \textit{no} implementation of $\choice$ is OSP.

\begin{proof}
  First we show that deciding OSP-implementability is in co-\textsf{NP}. By Corollary~\ref{corr:finite_rule_out}, in order to show that $\choice$ is not OSP-implementable, it is sufficient to present a non-singleton set $T'=T_1'\times\cdots\times T_n'\subseteq T$, along with a proof that $T'$ has no child in the O-dag of $\choice$. By definition $T'$ has no child iff for all $i$, the graph $G_i(T')$ from the proof of Lemma~\ref{lem:findchildren} is connected, or equivalently, $G_i(T')$ has a spanning tree. Recall that $G_i(T')$ has vertex set $T_i'$, and a pair $\{t_i,t_i'\}\subseteq T_i'$ forms an edge in $G_i(T')$ iff there exist some $t_{-i},t_{-i}'\in T_{-i}'$ such that one of the following inequalities holds:
  \ifnocount
  \par\textbf{Equation}\par
  \else
  \begin{align}
    \label{eq:edgecondition}
    u_i(t_i,\choice(t_i,t_{-i})) &< u_i(t_i,\choice(t_i',t_{-i}')) \\
    u_i(t_i',\choice(t_i',t_{-i}')) &< u_i(t_i',\choice(t_i,t_{-i})).
  \end{align}
  \fi
  Thus a proof that $\choice$ is not OSP-implementable may consist of the following:
  \begin{itemize}
  \item A vertex $T'$ with no child in the O-dag of $\choice$
  \item For each $i$, a set $H_i$ containing $|T_i'|-1$ tuples $(t,t')\in{T'}^2$ such that:
    \begin{enumerate}
    \item\label{it:span} $\{\{t_i,t_i'\}:(t,t')\in H_i\}$ forms a spanning tree of $G_i(T')$
    \item\label{it:edge} Each $(t,t')\in H_i$ satisfies at least one of the inequalities in~(\ref{eq:edgecondition}).
    \end{enumerate}
  \end{itemize}
  To verify the proof that each $G_i(T')$ is connected, it is sufficient to verify for each $i$ that $\{\{t_i,t_i'\}:(t,t')\in H_i\}$ forms a tree over the vertex set $T_i'$, and that each $(t,t')\in H_i$ satisfies at least one inequality in~(\ref{eq:edgecondition}). These verifications may be done in time $O(|H_i|)$ and $O(|T_i|(|C_{u_i}|+|C_{\choice}|))$ time respectively, where $|H_i|=(|T_i'|-1)\cdot 2\sum_i\lceil\log|T_i|\rceil$ is the number of bits to store $H_i$. The circuit $C_{\choice}$ takes as input $t\in T$, so its size $|C_{\choice}|$ is at least the number of input bits $\sum_i\lceil\log|T_i|\rceil$, and therefore $|H_i|=O(|T_i||C_{\choice}|)$. Thus the verification for all $i$ runs in time
  \begin{equation}
    \sum_iO(|H_i|+|T_i|(|C_{u_i}|+|C_{\choice}|))=O\left(\sum_i|T_i|(|C_{u_i}|+|C_{\choice}|)\right),
  \end{equation}
  which is at most quadratic in $|(1^{|T_1|},\dots,1^{|T_n|},C_{u_1},\dots,C_{u_n},C_{\choice})|$. Therefore deciding OSP-implementability is in co-\textsf{NP}.

  In the proof of \Cref{thm:versp}, we reduced SAT to the problem of verifying that a choice rule is not SP.  This may also be used to reduce SAT to verifying that no OSP mechanism exists for a choice rule. In particular, all choice rules used in the reduction have at most one agent $i$ for which $u_i$ is non-constant.  If such a choice rule is SP, then any mechanism in which all agents in $N\setminus\{i\}$ report their types and finally $i$ reports his type is OSP.  Hence, such a choice rule is SP iff it is OSP-implementable. Thus, deciding OSP-implementability is co-\textsf{NP}-complete.
\end{proof}

The preceding results suggest that SP and OSP are computationally comparable.  When inputs are expressed as a table, the SP and OSP decision problems are both in \textsf{P}, while when inputs are expressed as circuits, both decision problems are co-\textsf{NP}-complete.  Even though OSP depends on the extensive form of the mechanism, this does not seem to substantially raise computational complexity.

\section{Concluding remarks}

We have established that, given an arbitrary choice rule expressed as a table, one can determine whether it is OSP-implementable in sub-quadratic time.  Another approach is to take as given some economic setting, and then to investigate which choice rules are OSP-implementable.  This approach has been pursued successfully in a variety of settings, such as
\begin{itemize}
\item two-sided matching \citep{ashlagi_stable_2018, thomas2020classification},
\item social choice with single-peaked preferences \citep{bade_gibbard-satterthwaite_2017, arribillaga_obvious_2020},
\item object allocation \citep{troyan_obviously_2019, bade2019matching},
\item machine scheduling and facility location \citep{ferraioli2020approximation},
\item under `rich' preferences \citep{pycia_theory_2019},
\item and in the division problem with single-peaked preferences \citep{arribillaga2019all}.
\end{itemize}

These two approaches are distinct and complementary.  The algorithm presented above determines whether a single choice rule is OSP-implementable.  Of course, there are important questions that this does not directly answer.  For instance, given some class of choice rules that satisfy other desiderata, is every rule in the class OSP-implementable \citep{ashlagi_stable_2018, arribillaga2019all}?  As another example, if we restrict attention to some economic setting, do the OSP mechanisms have a special structure, such as clock auctions \citep{li_obviously_2017} or millipede games \citep{pycia_theory_2019}? 

We hope that an algorithm for rapidly checking individual cases will help to test and refine conjectures.  Notably, the O-dag characterization reveals that OSP-implementable choice rules have a tractable underlying structure.

\ifnocount
\makeatletter 
\renewcommand\BR@b@bibitem[2][]{\BR@bibitem[#1]{#2}\BR@c@bibitem{#2}}           
\makeatother
\bibliographystyle{ecta}
\nobibliography{library}
\else
\bibliographystyle{ecta}
\bibliography{library}
\fi

\end{document}